\documentclass{article}
\usepackage[english]{babel}
\usepackage[utf8]{inputenc}
\usepackage[bold,full]{complexity}

\usepackage{graphicx}
\usepackage{algorithm}
\usepackage[noend]{algpseudocode}

\title{$\cL$ is unequal $\NL$ under the Strong Exponential Time Hypothesis}
\makeatletter%
\@ifclassloaded{llncs}%
{\author{Reiner Czerwinski \orcidID{0000-0002-4523-4420}}
\institute{TU Berlin (Alumnus)\\
  \email{Reiner.Czerwinski@posteo.de}}
\titlerunning{$\cL\not=\NL$ under the SETH}
\authorrunning{R. Czerwinski}

}%
{\author{Reiner Czerwinsk}
  \usepackage{mypaper}
\usepackage[pdfstartview=FitH,pdfpagemode=UseNone,colorlinks=true,citecolor=blue,linkcolor=blue]{hyperref}}%
  
\begin{document}


\makeatother%

\maketitle
\ifdefined\text
\else
\newcommand{\text}[1]{\ensuremath{\textrm{#1}}}
\fi
\newlang{\stCON}{st-CON}
\newlang{\CNF}{CNF}
\newcommand{\kCNF}{\ensuremath{k-\CNF}}
\newcommand{\UNL}{\ensuremath{U_{\NL}}}
\newcommand{\UNP}{\ensuremath{U_{\NP}}}

  \begin{abstract}
     Due to Savitch's theorem we know $\NL\subseteq\DSPACE(\log^2(n))$.
    To show this upper bound, Savitch constructed an algorithm with
    $O(\log^2(n))$ space on the working tape. We will show that
    Savitch's algorithm also described a lower bound under the
    Strong Exponential Time Hypothesis. Every algorithm
    for the Connectivity Problem needs $O(\log^2(n))$ space in this case.
  \end{abstract}

\section{Introduction}
There are well-known upper bounds in the relation between nondeterministic
and deterministic complexity classes. It is obvious that
$\NSPACE(t)\subseteq\DTIME(2^{O(t)})$, and, due to Savitch's theorem\cite{SAVITCH},
$\NSPACE(s)\subseteq\DTIME(s^2)$.
In this paper, we analyze, whether these bounds are also lower bounds
under the Strong Exponential Time Hypothesis (SETH).

The proof of Savitch's theorem is constructive. There is an algorithm
to find a path between two vertices in a digraph, which needs $O(\log^2(n))$
space. To show that Savitch's algorithm is optimal, we use the following approach:
First, we will prove that
for each algorithm there is a graph with $n$ vertices where the algorithm needs space to
mark more than one vertex.
With this graph, we will construct a graph with $n^l$ vertices, where we have to mark $l$
vertices.

In section \ref{sec:ntm} we introduce nondeterministic TMs (NTM).
We need the NTM to analyze the effect of SETH on \NL-complete problems in section
 \ref{sec:SETH}. 

In section \ref{sec:savatich} , we will show that the algorithm described in Savitch's theorem has
minimal space complexity when SETH is valid. So, SETH implies $\cL\not=\NL$. 
\section{Nondeterministic Turing Machines} \label{sec:ntm}
In this paper, we use a nondeterministic Turing machine (NTM)
with an additional certificate tape\cite[p.84-86]{arora2009computational}.
An NTM has a read-only input tape, a read-once certificate tape,
one working tape, and an output tape.
W.l.o.g., the head of the certificate tape moves one cell to the right on each step of
the NTM.
An input $x$ is accepted by the NTM $M$, i.e $x\in L(M)$, if there is a certificate $c\in\{0,1\}^*$ 
with $M(x,c)=1$.

\includegraphics[width=10cm]{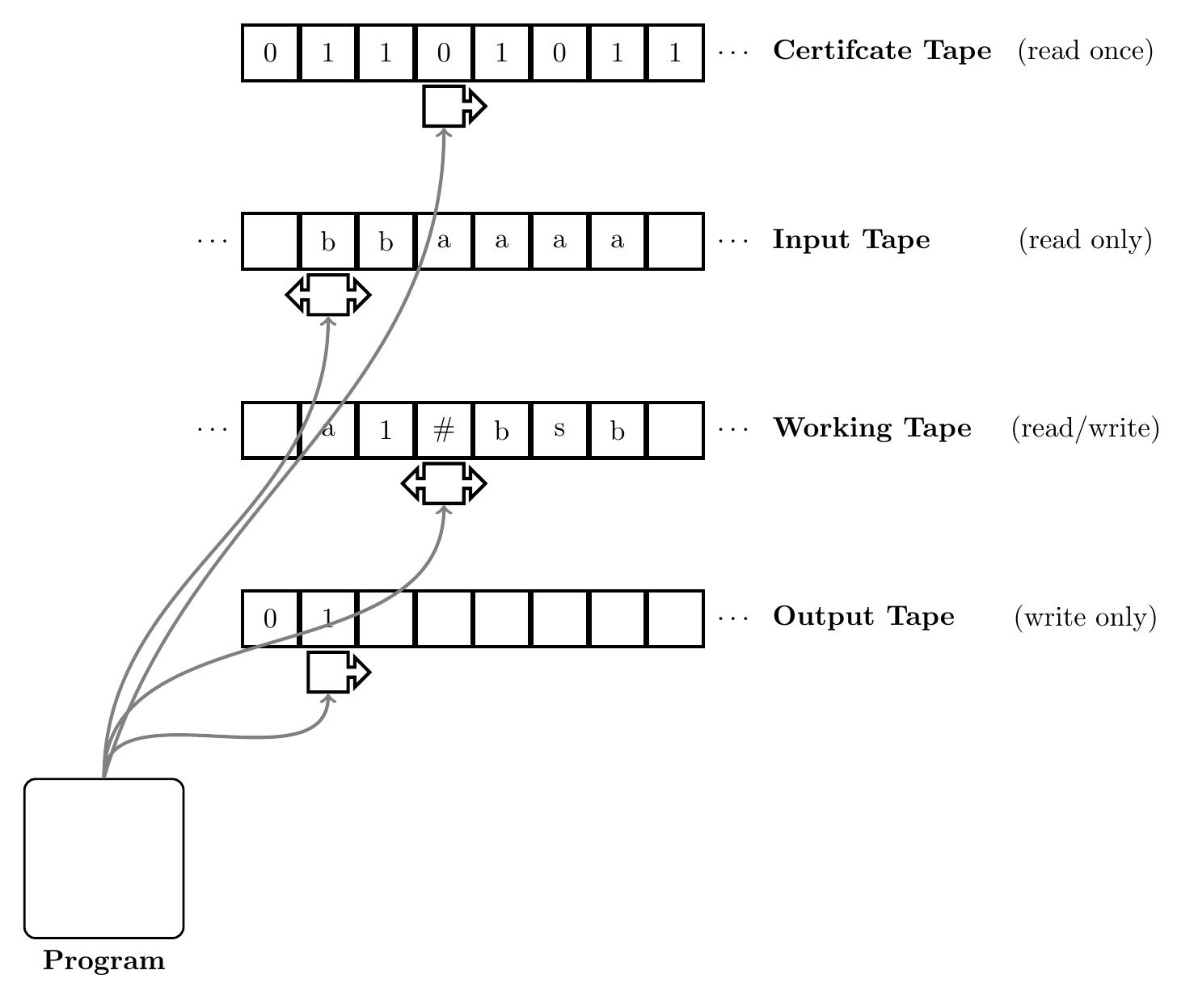}

The configuration is the information one needs to continue
the calculation of an NTM from a specific moment.
It is defined by the words on the input tape and the working tape
and the head position of the input and working tape and the state of the NTM.
The configuration does not depend on the content of the certificate tape.

The configuration of an NTM is denoted by the content of the working and output tape, and the position of the heads from the input, working, and output tape.

The configuration graph is a digraph where the vertices are
possible configurations of the NTM. Two configurations are connected with
a directed edge if one configuration can be transformed into the other
within one step of the NTM. If the out-degree of a vertex of the configuration
graph is more than one, the possible edges are distinguished by the
content of the current cell on the certificate tape. So,
the certificate describes a path in the configuration graph
but is not part of the configuration.




%
%


\section{ The Strong Exponential Time Hypothesis}\label{sec:SETH}

The Exponential Time Hypothesis implies that a deterministic algorithm for $\SAT$ needs exponential run time.
If the Strong Exponential Time Hypothesis (SETH) is valid, then an algorithm for an \SAT-formula in $\kCNF$
form with $n$ variables needs $O(2^n)$ run time in the worst case when $k\to\infty$\cite{SETH}.

The SETH is related to NTMs via the Cook--Levin theorem\cite{Cook71}.
We define
\begin{equation}
  \UNP = \{ \langle N,x,1^k  \rangle \mid \text{ NTM }N\text{ accepts }x\text{ within }k\text{ steps} \} 
\end{equation}
\begin{equation}
\UNL =\{ \langle N,x,1^k  \rangle \mid \text{ NTM }N\text{ accepts }x\text{ within }\log(k)\text{ working space} \}\text{.} 
\end{equation}

Obviously  the set \UNP{} is \NP-complete, and the set \UNL{} is \NL-complete.
The Cook--Levin theorem uses a p-many-one reduction from \UNP{} to $\SAT$. With this proof technique we can analyze \UNL{} under SETH.

\begin{lemma}\label{allpaths}
  Assume the Strong Exponential Time Hypothesis. 
 Then we have to check every possible path beginning with the start configuration of the configuration graph. 
\end{lemma}
\begin{proof}
  We use the reduction from the Cook-Levin theorem. The configuration at every step
  during the run time of the NTM is described by Boolean variables.
  The variable $T_W(t,p,s)$ is true if the symbol $s$ is on the working space tape on $p$ at step $t$.
  The variable $H_W(t,p)$ is true if the head position of the working space tape is on $p$ at step $t$.
  And there are similar variables for the other heads.
  
  Additionally, we get clauses from the space limitation.
  For this reason we introduce the Boolean variables $S_W(t,p)$, which is true when the head of the
  working space tape has passed the position $p$ before step $t$.
  It is $S_W(t+1,p) = S_W(t,p) \lor H_W(t,p)$. If $k$ is the maximal tape space, then
  we add the clauses$(\lnot S_W(t,p) \lor \lnot S_W(t,p+k) ) $ for all $t$,$p$.

  Assume SETH. If the NTM has not used the whole tape space, then an algorithm has to check all paths.
\end{proof}

Although we have to check all paths, an algorithm for $\UNL$ can be accelerated by
memorizing some configurations.


\section{Savitch's Algorithm is optimal}\label{sec:savatich}
An \NL-complete problem is $\stCON$.
Is there a path between two given points of a digraph?
This problem is in $\NL$ because an NTM needs only the working space size for one point,
i.e. $\log(n)$, where $n$ is the number of vertices.
The problem is \NL-hard because there is a reduction from $\UNL$ to $\stCON$
by writing the configuration graph.

An \stCON algorithm has to write some information about the digraph on the working space.
W.l.o.g., it can store information by marking some vertices.
The digraph consists of vertices and directed edges. Each edge 
is a tuple of two vertices.
If one marks a vertex by writing it on the working tape, then one needs 
$O(\log(n))$ space.

\begin{lemma}\label{twovertices}
  Assume SETH.
  For each algorithm that solves $\stCON$,
  there is a digraph where at least two vertices are marked at one moment.
\end{lemma}
\begin{proof}
  The problem $\stCON$ is $\NL$-complete. So, there is a reduction from
  $\UNL$ to $\stCON$.
  We construct the configuration graph with a logspace-transmuter.
  Due to theorem \ref{allpaths}, a deterministic algorithm has to test
  all possible paths from the start configuration to an accepted state
  in the
  worst case. If the algorithm marks only one vertex each time, it has to
  traverse every path. In this case, the run time is exponential. So, the working
  space grows faster than logarithmic.

  So, there is a graph where the algorithm has to mark more than one vertex
  to shorten the computation time.
\end{proof}

\begin{definition}[st-substitution]
  There is a digraph $G$ with vertex $v$ and a digraph
  $H$ with vertices $s_H$ and $t_H$. The vertex $v$ will be substituted
  by $H$ when every vertex adjacent to $v$ will be adjacent to $s_H$
  and  every vertex adjacent from $v$ will be adjacent from $s_H$.  
\end{definition}

When an algorithm for $\stCON$ would mark the substituted vertex in
the origin digraph, then it would solve $\stCON$ in the subgraph that
substituted the vertex.
So, we can generalize lemma \ref{twovertices}:
\begin{lemma} \label{morevertices}
  Assume SETH.
  For an arbitrary deterministic algorithm for $\stCON$,
  there is a digraph with the following property:
  If one st-substitutes each vertex with a subdigraph,
  then there is a moment during the run time
  where the algorithm solves $\stCON$ on one of these subdigraphs
  and a vertex outside this subdigraph is marked.
\end{lemma}
\begin{proof}
  As mentioned in lemma \ref{twovertices}, the number of paths,
  that an algorithm has to check can grow exponentially in
  relation to the number of vertices before the substitutions.
  So, there is a moment during the run time when the algorithm
  computes on a substituting subdigraph and has marked a vertex
  outside this subdigraph.
\end{proof}
Now we use the st-substitution to construct a digraph, where more
than two vertices have to be marked at some moment.

\begin{algorithm}
\caption{Generation of the graph $G^l$}
\begin{algorithmic}
  \State $K \gets G$
  \If {$l > 1$}
  \For {$v \in V(K)$}
  \State substitute $v$ with $G^{l-1}$
  \EndFor
  \EndIf
  \State $G^l \gets K$
\end{algorithmic}
\end{algorithm}
If the graph $G$ has $n$ vertices, then $G^l$ has $n^l$ vertices.
\begin{lemma}
  If SETH, then
  for an arbitrary algorithm to solve $\stCON$ there is a digraph $G$,
  where $A$ will mark at least $l+1$ vertices at one moment
  on input $G^l$.
\end{lemma}
\begin{proof}
  We proof this lemma with induction over $l$.
  Due to lemma \ref{twovertices}, there is a digraph $G$ where
  the algorithm $A$ has to mark at least two vertices at one moment.
  If $l=1$, then $G^l=G$. So, there is a moment at the run time of $A$
  where it has marked  at least two vertices.

  For the graph $G^{l+1}$ there is a digraph $K$ isomorphic to $G$,
  where every vertex will be substituted by $G^{l}$.
  If $v$ is a vertex of $K$, then $G^{l}(v)$ is the subdigraph,
  that will substitute $v$. So, an \stCON{} algorithm has to
  solve \stCON{} on the subdigraph $G^{l}(v)$ instead of marking
  $v$. By induction hypothesis one has to mark $l+1$ vertices
  of this subdigraph. But $K$ is isomorphic to $G$.
  Due to theorem \ref{morevertices}, there is
  a vertex $v$, where one also has to mark a vertex outside
  of $G^{l}(v)$ when calculating \stCON{} on this subdigraph.
\end{proof}
The graph $G^l$ has $N=n^l$ vertices, and the algorithm $A$ needs
working space to mark $l+1$ vertices. One vertex needs $\Theta(\log(N))$
space. So, we need $\Theta(\log^2(N))$ working space.
The SETH implies that Savitch's Algorithm is optimal.

\bibliographystyle{plane}
\bibliography{lit}

\end{document}